\pdfoutput=0
\documentclass[10pt,journal]{IEEEtran}
\usepackage{color, soul, framed}
\usepackage{enumerate}
\usepackage{amsfonts}
\usepackage{graphicx}
\usepackage{color}
\usepackage{amsmath,amsfonts,amssymb,amsthm,epsfig,epstopdf,url,array}
\usepackage{url,textcomp}
\usepackage{mdwmath}
\usepackage{mdwtab}
\usepackage{times}
\usepackage{authblk}
\usepackage{cite}
\usepackage{amsbsy}
\usepackage{mathtools}
\usepackage{subfigure}
\usepackage{algorithmic}
\usepackage{stfloats}
\usepackage[misc]{ifsym}
\newcommand{\bs}{\boldsymbol}
\DeclareMathAlphabet{\mathpzc}{OT1}{pzc}{m}{it}
\newtheorem{theorem}{Theorem}
\newtheorem{lemma}{Lemma}

\newtheorem{definition}{Definition}
\usepackage{bicaption}
\newcommand\blfootnote[1]{%
  \begingroup
  \renewcommand\thefootnote{}\footnote{#1}%
  \addtocounter{footnote}{-1}%
  \endgroup
}

\begin{document}
\title{Outage Performance Analysis of Full-Correlated Rayleigh MIMO Channels}
\author{\IEEEauthorblockN{Huan Zhang\IEEEauthorrefmark{1}, Guanghua~Yang\IEEEauthorrefmark{2}, Zheng~Shi$^{\textrm{\Letter}}$\IEEEauthorrefmark{2}, Shaodan~Ma\IEEEauthorrefmark{1}, and
Hong~Wang\IEEEauthorrefmark{3}\\
\vspace{-1em}
\IEEEauthorrefmark{1}Department of Electrical and Computer Engineering, University of Macau, China}\\
\IEEEauthorrefmark{2}School of Intelligent Systems Science and Engineering, Jinan University, Zhuhai, China\\
\IEEEauthorrefmark{3}School of Communication and Information Engineering, Nanjing University of Posts and Telecommunications, Najing, China
\vspace{-1.5em}}
\maketitle
\begin{abstract}
The outage performance of multiple-input multiple-output (MIMO) technique has received intensive attention to meet the stringent requirement of reliable communications  for 5G applications, e.g., mission-critical machine-type communication (cMTC). To account for spatial  correlation effects at both transmit and receive sides, the full-correlated Rayleigh MIMO fading channels are modeled according to Kronecker correlation structure in this paper. The outage probability is expressed as a weighted sum of the generalized Fox's H functions. The simple analytical result empowers asymptotic outage analysis at high signal-to-noise ratio (SNR), which not only reveal helpful insights into understanding the behavior of fading effects, but also offer useful design guideline for MIMO configurations. Particularly, the negative impact of the spatial correlation on the outage probability is revealed by using the concept of majorization, and the asymptotic outage probability is proved to be a monotonically increasing and convex function of the transmission rate. 
In the end, the analytical results are validated through extensive numerical experiments.
\end{abstract}
\begin{IEEEkeywords}
Asymptotic analysis, Mellin transform, MIMO, Outage probability, Spatial correlation.

\end{IEEEkeywords}
\IEEEpeerreviewmaketitle
\section{Introduction}
\blfootnote{The corresponding author is Zheng Shi. This work was supported in part by the National Key Research and Development Program of China under Grant 2017YFE0120600, in part by National Natural Science Foundation of China under Grants 61801192 and 61801246, in part by Guangdong Basic and Applied Basic Research Foundation under Grant 2019A1515012136, in part by the Science and Technology Planning Project of Guangdong Province under Grants 2018B010114002 and 2019B010137006, in part by the Science and Technology Development Fund, Macau SAR (File no. 0036/2019/A1 and File no. SKL-IOTSC2018-2020), in part by Open Research Foundation of National Mobile Communications Research Laboratory of Southeast University under Grant 2018D09, and in part by the Research Committee of University of Macau under Grant MYRG2018-00156-FST.}

\IEEEPARstart{5}{G} systems are anticipated to not only support extraordinarily high data rate and capacity, but also provide ultra-reliability, scalability and low latency for emerging communication paradigms, e.g., mission-critical machine-type communications (cMTC) \cite{Tullberg2016METIS}. The multiple-input multiple-output (MIMO) is an indispensable enabler that fulfills these requirements for 5G. Specifically, the MIMO technique capitalizes on the spatial dimension to explore its potentials of boosting the spectral efficiency and reliability \cite{telatar1999capacity}.
Most of the existing works concentrate on studying the information-theoretical capacity of MIMO systems for the purpose of the spectral efficiency enhancement\cite{shin2003capacity, hanlen2012capacity}.

Aside from the spectral efficiency, the reliability of communications has become of ever-increasing importance in the Internet-of-Things (IoT) applications (e.g., automated transportation, industrial control and augmented/virtual reality), ultra-reliable low-latency communications (URLLC) and tactile Internet \cite{avranas2018energy}. Since the outage probability is frequently used to characterize the reception reliability, the outage probability of MIMO systems also has attracted considerable attention in the literature \cite{telatar1999capacity,foschini1998limits,hochwald2004multi,
wang2004outage}. However, the prior works in \cite{telatar1999capacity,foschini1998limits,hochwald2004multi,wang2004outage} did not take into account the correlation between antenna elements, which exists in realistic propagation environments because of mutual antenna coupling and close spacing between adjacent elements \cite{shin2008mimo}. The spatial correlation would remarkably impair
 the reliability of MIMO systems. In \cite{khan2005capacity},  the character expansion method was initially introduced to give a closed-form expression for the moment-generating function (MGF) of the capacity under full-correlated (correlation at both the transmitter and receiver) Rayleigh MIMO channels if the numbers of transmit and receive antennas are identical. The same method was further extended to derive the MGF of the capacity for the case with arbitrary numbers of transmit and receive antennas in \cite{simon2006capacity}. Unfortunately, the outage probability for full-correlated Rayleigh MIMO systems was obtained in the literature by relying upon either approximations or numerical inversions, which impede the extraction of more helpful insights about the system parameters. Finally, the analytical results are verified by numerical analysis.

To address the above issues, The Mellin transform is applied in this paper to derive exact and tractable representations for the outage probabilities. Upon the exact expressions, the asymptotic analysis of the outage probability in the high signal-to-noise ratio (SNR) regime is derived. The expression demonstrates that full diversity can be achieved regardless of the presence of spatial correlation. The qualitative relationship between the spatial correlation and the outage probability is established by virtue of the concept of majorization in \cite{shin2008mimo}. Moreover, the transmission rate affects the outage performance via the term of modulation and coding gain, and the asymptotic outage probability is found to be an increasing and convex function of the transmission rate.\\
\emph{Notations:} We shall use the following notations throughout the paper. Bold uppercase and lowercase letters are used to denote matrices and vectors, respectively.  ${\bf A}^{\mathrm{H}}$, ${\bf A}^{-1}$ and ${\bf A}^{1/2}$ denote the conjugate transpose, matrix inverse and Hermitian square root of matrix ${\bf A}$, respectively. $\mathrm{vec}$, $\rm{tr}$, $\mathrm{det}$ and $\mathrm{diag}$ are the operators of vectorization, trace, determinant and diagonalization, respectively. $\Delta \left( {\bf A} \right)$ refers to the Vandermonde determinant of the eigenvalues of matrix ${\bf A}$. $\mathbf{0}_n$ and $\mathbf{I}_n$ stand for $1 \times n$ all-zero vector and $n \times n$ identity matrix, respectively. $\mathbb C^{m\times n}$ denotes the sets of $m\times n$-dimensional complex matrices. The symbol ${\rm i}=\sqrt{-1}$ is the imaginary unit. $o(\cdot)$ denotes little-O notation. $(\cdot)_n$ represents Pochhammer symbol. $|S|$ refers to the cardinality of set $S$. Any other notations will be defined in the place where they occur.
\vspace{-1em}
\section{System Model}\label{sec:sys_mod}
By considering a point-to-point MIMO system with $N_t$ transmit and ${N_r}$ receive antennas, the received signal vector $\mathbf{y}\in\mathbb{C}^{{N_r}\times 1}$ is written as
\begin{equation}\label{eq:pp_sig}
\mathbf{y}=\sqrt{\frac{P}{{N_t}}}\mathbf{H}\mathbf{x}+\mathbf{n},
\end{equation}
where $\mathbf{H}\in \mathbb{C}^{{N_r}\times {N_t}}$ is the matrix of the channel coefficients, $\mathbf{x}\in \mathbb{C}^{{N_t}\times 1}$ denotes the vector of transmitted signals, $\mathbf{n}\in \mathbb{C}^{{N_r}\times 1}$ represents the complex-valued additive white Gaussian noise vector with zero mean and covariance matrix $\sigma^{2}\mathbf{I}_{N_r}$, and $P$ is the total average transmitted power. Moreover, in order to account for the effect of the antenna correlation, the channel matrix $\mathbf{H}$ is modeled herein according to the Kronecker correlation structure as\cite{larsson2008space}
\begin{equation}\label{kron_mod}
\mathbf{H}={\mathbf{R}_r}^{{1}/{2}}{\bf H}_w{\mathbf{R}_{t}}^{{1}/{2}},
\end{equation}
where $\mathbf{H}_w\in \mathbb{C}^{{N_r}\times {N_t}}$ is a random matrix whose entries are independent and identically distributed (i.i.d.), complex circularly symmetric Gaussian random variables, i.e., $\mathrm{vec}\left(\mathbf{H}_w\right) \sim \mathcal{CN}\left(\mathbf 0_{{N_t}{N_r}},\mathbf I_{N_t}\otimes \mathbf I_{N_r}\right)$, $\mathbf{R}_t$ and $\mathbf{R}_{r}$ are respectively termed as the transmit and receive correlation matrices, and both of them are positive semi-definite Hermitian matrices.
For the sake of simplicity, we assume that the correlation matrices follow the constraints as ${\rm tr}{\left(\mathbf R_{t}\right)}=N_t$ and ${\rm tr}{\left(\mathbf R_{r}\right)}=N_r$.
From the perspective of information theory for MIMO system in \cite{simon2006capacity}, the outage probability can be expressed as
\begin{equation}\label{eqn:out_def}
\begin{aligned}
p_{out}&= \Pr \left( {{{\log }_2}\det \left( {{{\bf{I}}_{N_r}} + {\rho}{\bf{H}}{{\bf{H}}^{\rm{H}}}} \right) < R} \right) \\
&= \Pr \left( G\triangleq{\prod\limits_{i = 1}^{{N_r}} {\left( {1 + {\rho}{\lambda _i}} \right)} } < {2^R} \right) = F_G(2^R),
\end{aligned}
\end{equation}
where $\rho ={P}/{(\sigma^{2}N_t)}$ stands for the average transmit SNR per antenna, $\lambda_1,\cdots,\lambda_{N_r}$ denotes the unordered eigenvalues of ${{\bf{H}}}{\bf{H}}^{\rm{H}}$\footnote{It is worth noting that ${{\bf{H}}}{\bf{H}}^{\rm{H}}$ is a singular matrix and has at least $(N_r-N_t)$ zero eigenvalues if ${N_t}<{N_r}$.}, $F_G(x)$ denotes the cumulative distribution function (CDF) of $G$, with the property \cite[Exercise 7.25, p167]{abadir2005matrix}, the outage probabilities for $N_t \ge N_r$ and $N_t < N_r$ can be derived in the same fashion. Hence, we assume $N_t \ge N_r$ in the sequel unless otherwise specified. From \eqref{eqn:out_def}, it boils down to determining the distribution of the product of multiple shifted eigenvalues $\bs{\lambda }=(\lambda_1,\cdots,\lambda_{N_r})$.
\section{Analysis of Outage Probability}\label{sec:ana}
The outage probability is the fundamental performance metric to characterize the error performance of decodings. However, the occurrence of the correlation among eigenvalues will yield the involvement of a multi-fold integral in deriving the expression of $F_G(2^R)$. Nonetheless, the product form of $G$ motivates us to apply Mellin transform to obtain the distribution of $G$ \cite{shi2017asymptotic}. Specifically, the Mellin transform of the probability density function (PDF) of $G$, $\left\{ {\mathcal M f_G} \right\}\left( s \right)$, is given by
\begin{equation}\label{eqn:mellin_G}
\begin{aligned}
\varphi(s)=\int\nolimits_0^\infty  { \cdots \int\nolimits_0^\infty  {\prod\limits_{i = 1}^{N_r} {{{\left( {1 + {\rho}{\lambda _i}} \right)}^{s - 1}}f\left({\bs{\lambda }}\right)d\lambda_1  \cdots d{\lambda _{N_r}}} } }.
\end{aligned}
\end{equation}
By utilizing the inverse Mellin transform together with its associated property of integration  \cite[eq.(8.3.15)]{debnath2010integral}, the CDF of $G$ can be obtained as
\begin{equation}\label{eqn:cdf_g_inverse}
\begin{aligned}
{F_G}\left( x \right) &=  \left\{{\mathcal M} ^{ - 1}{\left[ - \frac{1}{s}\varphi \left( {s + 1} \right)\right]}\right\}\left( x \right)\\
&=\frac{{  1}}{{2\pi {\rm{i}}}}\int\nolimits_{c - {\rm i}\infty }^{c + {\rm i}\infty } {\frac{{{x^{ - s}}}}{-s}\varphi \left( {s + 1} \right)ds},
\end{aligned}
\end{equation}
where $c \in (-\infty,0)$, because the Mellin transform of ${F_G}\left( x \right)$ exists for any complex number $s$ in the fundamental strip $-\infty <\Re{(s)} < 0$ by noticing ${F_G}\left( x \right)=0$ for $x < 1$ and $\lim\nolimits_{x\to \infty}{F_G}\left( x \right)=1$ \cite[p400]{szpankowski2010average}.
\vspace{-0.8em}
\subsection{Exact Outage Probability}\label{sec:full}
By favor of the character expansions, the joint distribution of the $N_r$ unordered strictly positive eigenvalues of $\mathbf{H}\mathbf{H}^\mathrm{H}$ is obtained by Ghaderipoor \emph{et al.} in \cite{ghaderipoor2012application} as
\begin{equation}\label{eqn:eig_pdf_full}
f\left( {\bs{\lambda }}\right) = \sum\nolimits_{{{\bf{k}}_{N_r}}} {\frac{{{{\left( { - 1} \right)}^{\frac{{{N_r}\left( {{N_r} - 1} \right)}}{2}}}\mathcal A}}{{{N_r}!\Delta \left( {\bf{K}} \right)}}\Delta \left( {\bs{\lambda }} \right) \det \left( {{{\left\{ {{\lambda _i}^{{k_j} + {N_t} - {N_r}}} \right\}}_{ {i,j} }}} \right)},
\end{equation}
where $\mathcal A= \frac{{{\prod\nolimits_{i = 1}^{N_r} {{a_i}^{N_t}} \prod\nolimits_{j = 1}^{N_t} {{b_j}^{N_r}} }}}{{{\Delta \left( {\bf{A}} \right)\Delta \left( {\bf{B}} \right)}}{\prod\nolimits_{j = 1}^{N_r} {\left( {{k_j} + {N_t} - {N_r}} \right)!} }}\det {( \{{{\left( { - {a_i}} \right)}^{{k_j}}}\}_{i,j} )}\times \det ( {{{\{ {{b_i}^{{k_j} + N - M}} \}}_{{i,1 \le j \le M}}},{{\{ {{b_i}^{N - j}} \}}_{ {i,M + 1 \le j \le N}}}})$, ${\bf a} =(a_1,\cdots,a_{N_r})$ and ${\bf b}  =(b_1,\cdots,b_{N_t})$ represent the eigenvalue vectors of ${{\bf R}_r}^{-1}$ and ${{\bf R}_t}^{-1}$, respectively, ${\bf k}_{N_r}=(k_1,\cdots,k_{N_r})$ stands for all irreducible representation of the general linear group $\mathrm{GL}({N_r},\mathbb C)$ and $k_1\ge\cdots\ge k_{N_r}$ are integers, ${\bf A}$, $\bf B$ and $\bf K$ are the diagonalizations of vectors $\bf a$, $\bf b$ and ${\bf k}_{N_r}$, respectively.

By substituting (\ref{eqn:eig_pdf_full}) into (\ref{eqn:mellin_G}), the Mellin transform of $f_G(x)$ is expressed as
\begin{equation}\label{eqn:mellin_fullc_tricomi}
\begin{aligned}
\varphi\left( s \right) &= \frac{{{{\left( { - 1} \right)}^{{N_r}\left( {{N_t} - {N_r}} \right)}}{{{{\rho }} }^{-\frac{1}{2}{{{N_r}\left( {{N_r} + 1} \right)}}}}\prod\nolimits_{i = 1}^{N_r} {{a_i}^{N_r}} \prod\nolimits_{j = 1}^{N_t} {{b_j}^{N_r}} }}{{\Delta \left( {\bf{A}} \right)\Delta \left( {\bf{B}} \right)\prod\nolimits_{i = 1}^{N_r} {{{\left( {s + i - 2} \right)}^{i - 1}}} }}\\
&\times\det \left( {\begin{array}{*{20}{c}}
{{{\left\{ {\Psi \left( {1,s + {N_r};\frac{{{a_i}{b_j}}}{\rho }} \right)} \right\}}_{1 \le i \le {N_r},j}}}\\
{{{\left\{ {{b_j}^{{N_t} - i}} \right\}}_{{N_r} + 1 \le i \le {N_t},j}}}
\end{array}} \right).
\end{aligned}
\end{equation}
where $\Psi \left( \cdot,\cdot;\cdot \right)$ denotes  confluent hypergeometric function \cite[eq. (9.210)]{gradshteyn1965table}.
\begin{proof}
Please see Appendix \ref{app:full_mellin}
\end{proof}
By substituting \eqref{eqn:mellin_fullc_tricomi} into \eqref{eqn:cdf_g_inverse}, the CDF of $G$ can then be obtained as shown in the following theorem.
\begin{theorem}\label{the:full_cdf}
The CDF of $G$ is given by
\begin{equation}
\begin{aligned}\label{eqn:F_G_CDF_fullcorr}
{F_G}\left( x \right) &= \Lambda{{\rho  }^{ \frac{1}{2}{{{N_r}\left( {{N_r} - 1} \right)}}}}\sum\limits_{{\bs{\sigma }} \in {S_{N_t}}}{\mathop{\rm sgn}} \left( {\bs{\sigma }} \right)\times \\
&{\prod\limits_{i = {N_r} + 1}^{N_t} {{b_{{\sigma _i}}}^{{N_t} + {N_r} - i}} }\underbrace{Y_{{N_r},2{N_r}}^{{N_r},{N_r}}\left[ {\left. {\begin{array}{*{20}{c}}
\mathbf{C}\\
\mathbf{D}
\end{array}} \right|x {\prod\limits_{i = 1}^{N_r} {\frac{{{a_i}{b_{{\sigma _i}}}}}{\rho }} }} \right]}_{\mathcal Y_{ {{\bs\sigma}}}(x)},
\end{aligned}
\end{equation}
where $\Lambda=\frac{{{{\left( { - 1} \right)}^{{N_r}\left( {{N_t} - {N_r}} \right) + \frac{1}{2}{{{N_r}\left( {{N_r} - 1} \right)}}}}}}{{\Delta \left( {\bf{A}} \right)\Delta \left( {\bf{B}} \right)}}$, ${\mathcal Y_{ {{\bs\sigma}}}(x)}$ is the generalized Fox's H function which is defined by using the integral of Mellin-Branes type and provided by an efficient MATHEMATICA implementation as \cite{yilmaz2010outage}, $\mathbf{C}=[{\left( {1,1,0,1} \right),{{\left( {{N_r} ,1,0,1} \right)}_{j = 1, \cdots ,{N_r}-1}}}]$, and $\mathbf{D}=[{{{( {{N_r},1,\frac{{{a_{{i}}}{b_{\sigma_i}}}}{\rho },1})}_{i = 1, \cdots ,{N_r}}},\left( {0,1,0,1} \right),{{\left( {j,1,0,1} \right)}_{j = 1, \cdots ,{N_r}-1}}}]$.
\begin{proof}
The proof is given in Appendix \ref{app:full_cdf}.
\end{proof}
\end{theorem}

Substituting \eqref{eqn:F_G_CDF_fullcorr} into  \eqref{eqn:out_def}, the outage probability under fully correlated Rayleigh MIMO channels can be obtained.

The generalized Fox's H function involved in \eqref{eqn:F_G_CDF_fullcorr}  is too complex to extract insightful results. In order to obtain tractable results and gain more insights, we have to recourse to the asymptotic analysis of the outage probability at high SNR.
\vspace{-0.8em}
\subsection{Asymptotic Outage Probability}
In order to derive the asymptotic expression at high SNR for the outage probability in \eqref{eqn:F_G_CDF_fullcorr}, the following lemma associated with the asymptotic expression of  ${\mathcal Y_{ {{\bs\sigma}}}(x)}$ is developed first.
\begin{lemma}\label{the:asy_y_3}
As $\rho \to \infty$, ${\mathcal Y_{ {{\bs\sigma}}}(x)}$ is asymptotic to
\begin{equation}
\begin{aligned}\label{eqn:mellin_cdf_g_fullcsers_exp}
{\mathcal Y_{ {{\bs\sigma}}}(x)} &= {{{{\rho }}}^{-{N_r}^2}}\prod\limits_{i = 1}^{N_r} {{a_i}^{N_r}}{{b_{\sigma_i}}^{N_r}}\frac{1}{{2\pi {\rm{i}}}}\int\nolimits_{-c - {\rm i}\infty }^{-c + {\rm i}\infty } \frac{{\Gamma \left( s \right)}}{{\Gamma \left( {1 + s} \right)}}\\
 &\times\prod\limits_{i = 1}^{N_r} {\frac{{\Gamma \left( {s - {N_r}} \right)}}{{\Gamma \left( {s - i + 1} \right)}}}  {\sum\limits_{{n_{{i}}} = 0}^\infty  {\frac{{{{\left( {\frac{{{a_{{ i}}}{b_{\sigma_i}}}}{\rho }} \right)}^{{n_{{i}}}}}}}{{{{\left( { 1 + {N_r}  - s} \right)}_{{n_{{i}}}}}}}} {x^s}ds}\\
&+ o\left( {{\rho ^{ - {N_t}{N_r}-\frac{1}{2}N_r(N_r+1)}}} \right).
\end{aligned}
\end{equation}
\end{lemma}
\begin{proof}
Please see Appendix \ref{app:asy_y_3}.
\end{proof}

By using Lemma \ref{the:asy_y_3}, the asymptotic expression of the outage probability is given by the following theorem.
\begin{theorem}\label{the:asy_out_full}
At high transmit SNR, the outage probability asymptotically equals
\begin{align}\label{eqn:F_G_mellin_trans_detzerofina}
p_{out} = \frac{{{{{{\rho }} }^{-{N_r}{N_t}}}}}{{\det \left( {{{{\bf{R}}_r}}} \right)^{N_t}\det \left( {{{{\bf{R}}_t}}} \right)^{N_r}}}g(R)+ o\left( {{\rho ^{ - {N_t}{N_r}}}} \right),
\end{align}
where \resizebox{0.7\hsize}{!}{$g(R) = G_{{N_r} + 1,{N_r} + 1}^{0,{N_r} + 1}\left( {\left. {\begin{array}{*{20}{c}}
{1,{N_t} + 1 , \cdots ,{N_t} + {N_r}}\\
{0,1, \cdots ,{N_r}}
\end{array}} \right|2^R} \right)$}
 denotes the Meijer G-function\cite{gradshteyn1965table}.
\end{theorem}
\begin{proof}
 Please see Appendix \ref{app:asy_out_full}.
\end{proof}
It is worth mentioning that the asymptotic analysis of the outage probability for $N_t < N_r$ can be carried out in an analogous way owing to the property \cite[Exercise 7.25, p167]{abadir2005matrix}. Similar to \eqref{eqn:F_G_CDF_fullcorr} and \eqref{eqn:F_G_mellin_trans_detzerofina}, the exact and asymptotic outage expressions for $N_t < N_r$ can be obtained by directly interchanging ${\bf R}_t$ and $N_t$ with ${\bf R}_r$ and $N_r$, respectively.
\section{Discussions of Asymptotic Results}\label{sec:asym_res}
The asymptotic outage probabilities commonly exhibit the same basic mathematical structure as \cite[eq.(3.158)]{tse2005fundamentals} 
\begin{equation}\label{eqn:gen_out_asy}
p_{out}  =  \mathcal S ({\bf R}_t,{\bf R}_r) {\left( {{\mathcal C(R)} } \rho \right)^{ - d}} + o\left( {{\rho ^{ - d}}} \right),
\end{equation}
where $\mathcal S({\bf R}_t,{\bf R}_r)$ quantifies the impact of spatial correlation at transmit and receive sides, $\mathcal C(R)$ is the modulation and coding gain, and $d$ stands for the diversity order. By identifying \eqref{eqn:F_G_mellin_trans_detzerofina} with \eqref{eqn:gen_out_asy}, we get
\begin{equation}\label{eqn:spa_imp_exp}
\mathcal S({\bf R}_t,{\bf R}_r) = \frac{1}{{\det \left( {{{{\bf{R}}_r}}} \right)^{N_t}\det \left( {{{{\bf{R}}_t}}} \right)^{N_r}}},
\end{equation}
\vspace{-1em}
\begin{equation}\label{eqn:cm_imp_exp}
\mathcal C(R) = {\left({g_{\bf{0}}}(R)\right)}^{-\frac{1}{N_tN_r}},
\end{equation}
and $d=N_tN_r$, respectively.
%
To comprehensively understand the asymptotic behavior of the outage probability, these impact factors are discussed individually.
\subsection{Diversity Order}
The terminology of the diversity order can be used to measure the degree of freedom of communication systems, which is defined as the ratio of the outage probability to the transmit SNR on a log-log scale as
\begin{equation}\label{eqn:d_def}
 d = \mathop {\lim }\limits_{\rho  \to \infty } \frac{{\log {p_{out}}}}{{\log \rho }}.
\end{equation}
Hence, the diversity order indicates the decaying speed of the outage probability with respect to the transmit SNR. As shown in \eqref{eqn:F_G_mellin_trans_detzerofina},
full diversity can be achieved by MIMO systems regardless of the presence of the spatial correlation, i.e., $d=N_tN_r$.
\vspace{-0.8em}
\subsection{Modulation and Coding Gain}\label{sec:mc_gain}
The modulation and coding gain $\mathcal C(R)$ quantifies the amount of the SNR reduction required to reach the same outage probability when employing a certain modulation and coding scheme (MCS).
Accordingly, the increase of $\mathcal C(R)$ is in favor of the improvement of the outage performance. From \eqref{eqn:cm_imp_exp}, in order to disclose the behaviour of $\mathcal C(R)$, it suffices to investigate the property of
the function ${g_{\bf{0}}}(R)$. Towards this end, we can arrive at following Theorem.
\begin{theorem}\label{the:conv_g0}
${g_{\bf{0}}}(R)$ is a monotonically increasing and convex function of the transmission rate $R$.
\end{theorem}
\begin{proof}
Similar to \cite[Lemma 4]{shi2017asymptotic}, the detailed proof is omitted here to save space.
\end{proof}
Evidently from Theorem \ref{the:conv_g0}, the transmission rate is an increasing and convex function of the asymptotic outage probability. Without dispute, the monotonicity and convexity of ${g_{\bf{0}}}(R)$ can greatly facilitate the optimal rate selection of MIMO systems if the asymptotic results are used.
\vspace{-0.8em}
\subsection{Spatial Correlation}\label{sec:scorr}
Although the effect of the spatial correlation can be quantified by $\mathcal S({\bf R}_t,{\bf R}_r)$, it is also imperative to draw a qualitative conclusion about the outage behaviour of the spatial correlation.
To characterize the spatial correlation, the majorization theory is usually adopted as a powerful mathematical tool to establish a tractable framework \cite{shin2008mimo,feng2018impact}. The majorization-based correlation model is defined as follows.

\begin{definition}\label{Definition1}
For two $N \times N$ semidefinite positive matrices $\mathbf{R}_{1}$ and $\mathbf{R}_{2}$, ${\bf r}_1=(r_{1,1},\cdots,r_{1,{N}})$ and ${\bf r}_2=(r_{2,1},\cdots,r_{2,{N}})$ are defined as the vectors of the eigenvalues of $\mathbf{R}_{1}$ and $\mathbf{R}_{2}$, respectively, where the eigenvalues are arranged in descending order as ${r_{i,1}} \ge \cdots  \ge {r_{i,{N}}},i \in \left\{ {1,2} \right\}$. We denote $\mathbf{R}_{1} \preceq \mathbf{R}_{2}$ and say the matrix $\mathbf{R}_{1} $ is majorized by the matrix $ \mathbf{R}_{2}$ if

\begin{equation}
\resizebox{0.85\hsize}{!}{$
\sum\limits_{j = 1}^k {{r_{1,j}}}  \le \sum\limits_{j = 1}^k {{r_{2,j}}},\,\left( {k\leq{N} - 1} \right)  \quad{\rm and}\quad \sum\limits_{j = 1}^{N} {{r_{1,j}}}  = \sum\limits_{j = 1}^{N} {{r_{2,j}}} .$}
\end{equation}
We also say ${\mathbf{R}}_{1} $ is more correlated than $ \mathbf{R}_{2}$.
\end{definition}
It is easily found by definition that ${\rm diag}(1,1,\cdots,1) \preceq \mathbf{R}_i \preceq {\rm diag}({N},0,\cdots,0)$ if ${\rm tr}(\mathbf{R}_i) = N$, where ${\rm diag}({N_t},0,\cdots,0)$ and ${\rm diag}(1,1,\cdots,1)$ correspond to completely correlated and independent cases, respectively. Notice that $\det(\mathbf{R}_{i})=\prod\nolimits_{j=1}^{N}r_{i,j}$, the property of the majorization in \cite[F.1.a]{marshall1979inequalities} proves that the determinant of the correlation matrix is a Schur-concave function, where $\det(\mathbf{R}_{1}) \ge \det(\mathbf{R}_{2})$ if $\mathbf{R}_{1} \preceq \mathbf{R}_{2}$, the interested reader is referred to \cite{marshall1979inequalities} for further details regarding the schur monotonicity. By recalling $\mathcal S({\bf R}_t,{\bf R}_r)$ is the composition of the determinants and using the fact associated with the composition involving Schur-concave functions \cite{marshall1979inequalities}, we arrive at
\begin{small}
\begin{equation}\label{eqn:}
\mathcal S({\bf R}_{t_1},{\bf R}_{r_1}) \le \mathcal S({\bf R}_{t_2},{\bf R}_{r_2}),
\end{equation}
\end{small}
whenever ${\bf R}_{t_1} \preceq {\bf R}_{t_2}$  and $ {\bf R}_{r_1} \preceq {\bf R}_{r_2}$. As a consequence, it is concluded that the presence of the spatial correlation adversely impacts the outage performance.
\section{Numerical Results}\label{sec:num}
In this section, numerical results are presented for verifications and discussions. For notational convenience, we define $\mathbf t$ and $\mathbf r$ as the row vectors of the eigenvalues of the transmit and receive correlation matrices, i.e., $\mathbf R_t$ and $\mathbf R_r$.
\vspace{-0.8em}
\subsection{Verifications}
Figs. \ref{fig:out_Nt3}  depicts the outage probability versus the transmit SNR. The labels `Sim.', `Exa.' and `Asy.' in this figure indicates the simulated, exact and asymptotic outage probabilities, respectively. As observed in fig. \ref{fig:out_Nt3}, the exact and simulation results are in perfect agreement, which confirms the correctness of the exact analysis. Besides, it can be seen  that the asymptotic results coincide well with the exact and simulation ones at high SNR, which validates the asymptotic results as well.
\begin{figure}[!t]
\centering
\includegraphics[width=2.8in]{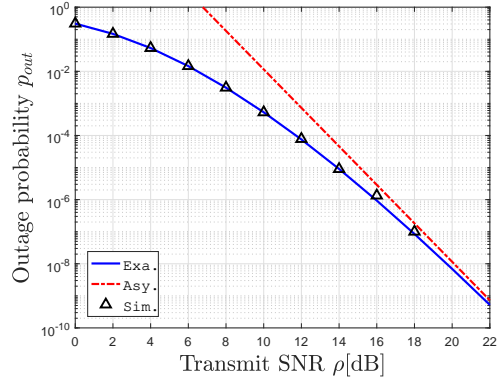}
\caption{Outage probability versus the transmit SNR $\rho$ with ${N_t}=3$, ${N_r}=2$, ${\mathbf t}=(2.7,0.2,0.1)$, ${\mathbf r}=(1.9, 0.1)$, and $R=$2bps/Hz.}\label{fig:out_Nt3}
\end{figure}
\vspace{-0.8em}
\subsection{Coding and Modulation Gain}
Fig. \ref{fig:code_gain} illustrates the impacts of the transmission rate $R$ on the coding and modulation gain ${\mathcal C(R)}$ for different numbers of transmit and receive antennas. It is shown in Fig. \ref{fig:code_gain} that ${\mathcal C(R)}$ increases with the number of antennas, which further justifies the benefit of using MIMO. Additionally, it can be observed from Fig. \ref{fig:code_gain} that the increase of the transmission rate impairs the coding and modulation gain ${\mathcal C(R)}$, which consequently leads to the deterioration of the outage performance. This is consistent with the asymptotic analysis in Section \ref{sec:mc_gain}. Aside from degrading the outage performance, the increase of the transmission rate causes the enhancement of the system throughput. The two opposite effects force us to properly select the transmission rate in practice. Fortunately, the optimal rate selection can be eased by using the asymptotic outage probability thanks to its increasing monotonicity and convexity with respect to the transmission rate.
\begin{figure}[!t]
\begin{center}
\includegraphics[width=2.8in]{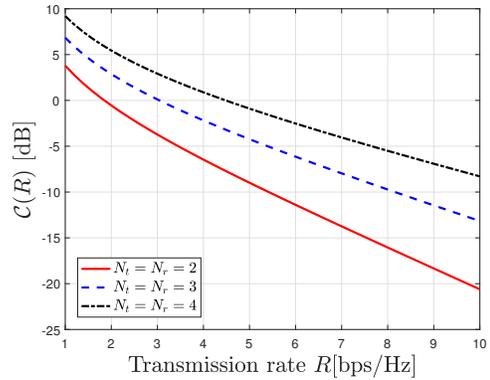}
\caption{Coding and modulation gain ${\mathcal C(R)}$ versus the transmission rate $R$.}\label{fig:code_gain}
\end{center}
\end{figure}
\vspace{-0.8em}
\subsection{Impact of Spatial Correlation}
In Fig. \ref{fig:corr_imp}, the outage probability is plotted against the transmit SNR under three different transmit correlation matrices, i.e., ${\bf{R}}_{t_1}$, ${\bf{R}}_{t_2}$ and ${\bf{R}}_{t_3}$. For notational simplicity, the vectors of the eigenvalues are set as ${\mathbf t}_1 = (1,1,1)$, ${\mathbf t}_2=(2.3,0.5,0.2)$ and ${\mathbf t}_3 = (2.7,0.2,0.1)$. It is clear from the figure that the spatial correlation does not affect the diversity order. Moreover, according to the concept of majorization, the relationship of the transmit correlation matrices follows as ${{\bf{R}}_{t_3}}\succeq{{\bf{R}}_{t_2}}\succeq{{\bf{R}}_{t_1}}$, and $\mathbf{R}_{t_3}$ are the most correlated correlation matrix among them. It is readily observed in Fig. \ref{fig:corr_imp} that the outage probability curve associated with $\mathbf{R}_{t_3}$ displays the worst performance. The numerical result corroborates the validity of the analytical results in Section \ref{sec:scorr}.
\begin{figure}[!t]
\begin{center}
\includegraphics[width=2.8in]{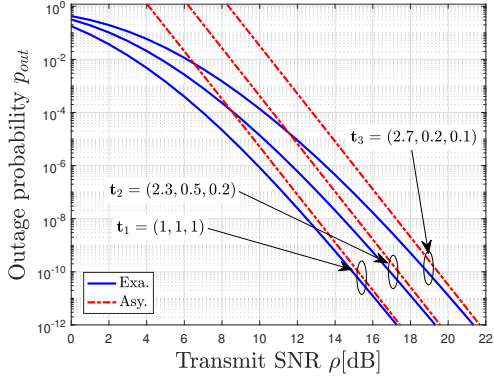}
\caption{Outage probability versus the transmit SNR $\rho$ with ${N_t}={N_r}=3$ and ${\mathbf r}=(2.7,0.2,0.1)$.}\label{fig:corr_imp}
\end{center}
\end{figure}
\vspace{-1em}
\section{Conclusions}\label{sec:con}
This paper has derived novel representation for the outage probabilitiy of the MIMO system by invoking Mellin transform. The compact and simple expression not only has enabled the accurate evaluation of the outage probability, but also has facilitated the asymptotic analysis under high SNR to gain a profound understanding of fading effects and MIMO configurations, which has never been performed in the literature. On the one hand, the asymptotic results have revealed meaningful insights into the effects of the spatial correlation, the number of antennas,  and transmission rate. For instance, the spatial correlation degenerates the outage performance, while full diversity can be achieved even at the presence of the spatial correlation. On the other hand, the asymptotic result have paved the way for the simplification of practical system designs. For example, the increasing monotonicity and convexity of the asymptotic outage probability will facilitate the proper selection of target transmission rate.
\vspace{-1em}
\appendices
\section{}\label{app:full_mellin}
By substituting \eqref{eqn:eig_pdf_full} into \eqref{eqn:mellin_G}, using the identity $\Delta \left( {\bs{\lambda }} \right) = \det \left( {{{\left\{ {{\lambda _i}^{j-1}} \right\}}_{{i,j}}}} \right)$ and the Leibniz formula for the determinant expansion \cite{horn2012matrix}, we can obtain that
\begin{equation}\label{eqn:delta_lambda}
\resizebox{0.9\hsize}{!}{$
\begin{aligned}
\varphi\left( s \right) &=\sum\nolimits_{{{\bf{k}}_{N_r}}} \frac{{{{\left( { - 1} \right)}^{\frac{{{N_r}\left( {{N_r} - 1} \right)}}{2}}}{\cal A}}}{{{N_r}!\Delta \left( {\bf{K}} \right)}}\sum\limits_{{{\bs{\sigma }}_1},{{\bs{\sigma }}_2} \in {S_{N_r}}} {\mathop{\rm sgn}} \left( {{{\bs{\sigma }}_1}} \right){\mathop{\rm sgn}} \left( {{{\bs{\sigma }}_2}} \right) \\
&\times\prod\limits_{i = 1}^{N_r} {\int\nolimits_0^\infty  {{{\left( {1 + \rho\lambda } \right)}^{s + N_r - {\sigma _{1,i}} - 1}}{\lambda ^{{k_{{\sigma _{2,i}}}} + {\sigma _{1,i}} + \tau}}d\lambda } }.
\end{aligned}$}
\end{equation}
where $S_{N_r}$ denotes the set of permutations of $\left\{1,2,\cdots,{N_r}\right\}$, $\bs \sigma_l \triangleq (\sigma_{l,1},\cdots,\sigma_{l,{N_r}})$ for $ l\in\left(1,2\right)$ and ${\rm sgn}(\bs \sigma_l)$ denotes the signature of the permutation $\bs \sigma_l$, i.e., ${\rm sgn}(\bs \sigma_l)$ is $1$ whenever the minimum number of transpositions necessary to reorder $\bs \sigma_l$ as $(1,2,\cdots,{N_r})$ is even, and $-1$ otherwise  \cite{simon2006capacity}.

\begin{lemma}\label{the:leb_for}
If $\eta ({\bs\sigma _1},{\bs\sigma _2})$ is a function of $\bs\sigma _1$ and ${\bs\sigma _2}$ irrespective of the ordering of the elements in the permutations of the set of two-tuples $\{(\sigma _{1,l},\sigma _{2,l}):l\in[1,N_r]\}$, the summation of ${{\mathop{\rm sgn}} \left( {{\bs\sigma _1}} \right){\mathop{\rm sgn}} \left( {{\bs\sigma _2}} \right)}\eta({\bs\sigma _1},{\bs\sigma _2})$ over all permutations of $\bs\sigma _1$ and $\bs\sigma _2$ degenerates to
\begin{equation}\label{eqn:leb_for_gen}
\begin{aligned}
&\sum\nolimits_{{{\bs\sigma _1},{\bs\sigma _2} \in {S_{N_r}}}} {{\mathop{\rm sgn}} \left( {{\bs\sigma _1}} \right){\mathop{\rm sgn}} \left( {{\bs\sigma _2}} \right)}\eta({\bs\sigma _1},{\bs\sigma _2})\\
&= \sum\nolimits_{{{\bs\sigma _1},{\bs\sigma _2} \in {S_{N_r}}}} { {{\rm{sgn}}\left( {{{\bs\sigma }}} \right){\mathop{\rm sgn}} \left( {\bar{\bs\sigma}} \right)\eta({\bs\sigma},\bar{\bs\sigma})} }\\
&= N_r!\sum\nolimits_{{{\bs\sigma } \in {S_{N_r}}}} {{\mathop{\rm sgn}} \left( {{\bs\sigma }} \right)}\eta(\bar{\bs\sigma},{\bs\sigma }),
\end{aligned}
\end{equation}
where ${\bs\sigma}=(\sigma_1,\cdots,\sigma_{N_r})$ and $\bar{\bs\sigma}=(1,\cdots,N_r)$.
\end{lemma}
\begin{proof}
Similar to the proofs of Leibniz formulae \cite[eq.(2)]{simon2006capacity} and \cite[eqs.(64-65)]{ghaderipoor2012application}, denote by $S_{N_r} \times S_{N_r}$ the Cartesian product of $S_{N_r}$. Hence, $(\bs \sigma_1, \bs \sigma_2)\in S_{N_r} \times S_{N_r}$. We further establish the one-to-one mapping $\vartheta(\bs \sigma_1, \bs \sigma_2)$ as a vector of ordered pairs $(\sigma _{1,l},\sigma _{2,l})$ for $l\in [1,N_r]$, i.e., $\vartheta(\bs \sigma_1, \bs \sigma_2) \triangleq \left((\sigma _{1,l},\sigma _{2,l}):l\in[1,N_r]\right)$. We thus reach the relation ${\mathop{\rm sgn}} \left( {{\bs\sigma _1}} \right){\mathop{\rm sgn}} \left( {{\bs\sigma _2}} \right) = {\mathop{\rm sgn}} \left( {{\bs\sigma }} \right){\mathop{\rm sgn}} \left( {\bar{\bs\sigma}} \right)$ after a certain number of transpositions to achieve $\vartheta(\bs \sigma, \bar{\bs\sigma})$ starting from $\vartheta(\bs \sigma_1, \bs \sigma_2)$, where ${\mathop{\rm sgn}} \left( {\bar{\bs\sigma}} \right)=1$. Secondary, according to the definition of $\eta ({\bs\sigma _1},{\bs\sigma _2})$, the one-to-one mapping and the cardinality of $S_{N_r}$, i.e., $|S_{N_r}|=N_r!$. Using the same manner the second step holds.
\end{proof}
Basing on Lemma \ref{the:leb_for} and the Leibniz formula for the determinant expansion \cite{horn2012matrix}, (\ref{eqn:delta_lambda}) can be simplified as
\begin{align}\label{eqn:mellin_trans_full_leibniz}
\varphi\left( s \right) = \sum\nolimits_{{{\bf{k}}_{N_r}}} {\frac{{{{\left( { - 1} \right)}^{\frac{{{N_r}\left( {{N_r} - 1} \right)}}{2}}}{\cal A}}}{{\Delta \left( {\bf{K}} \right)}}}{\xi \left( s \right)},
\end{align}
where ${\xi \left( s \right)}=\det [ {\{ {\int_0^\infty  {{{\left( {1 + \rho \lambda } \right)}^{s + {N_r} - i - 1}}{\lambda ^{{k_j} + i + \tau }}d\lambda } } \}}_{i,j} ]$. By the change of variable $x = 1/{{(1 + \rho \lambda )}}$, $\xi \left( s \right)$ can be expressed in terms of Beta function \cite[eq.(8.380.1)]{gradshteyn1965table} as
\begin{equation}\label{eqn:beta_fun}
 \xi \left( s \right)= {\rho ^{ - {k_j} - i - \tau  - 1}}{\rm{B}}\left( { - s - {N_r} - {k_j} - \tau ,{k_j} + i + \tau  + 1} \right).
\end{equation}
By using the relationship ${\rm B}(\alpha,\beta) = \Gamma(\alpha)\Gamma(\beta)/\Gamma(\alpha+\beta)$  \cite[eq.(8.384.1)]{gradshteyn1965table}
and the generalized Cauchy-Binet formula \cite[Lemma 4]{ghaderipoor2012application}, \eqref{eqn:mellin_trans_full_leibniz} can be simplified as
\begin{equation}\label{eqn:mel_fc_defa_sub}
\resizebox{0.8\hsize}{!}{$\begin{aligned}
\varphi\left( s \right)
& = \frac{{{{\left( { - 1} \right)}^{{N_r}\left( {{N_t} - {N_r}} \right)}}{\rho ^{ - \frac{1}{2}{N_r}\left( {{N_r} + 1} \right)}}\prod\nolimits_{i = 1}^{{N_r}} {{a_i}^{{N_r}}} \prod\nolimits_{j = 1}^{{N_t}} {{b_j}^{{N_r}}} }}{{\Delta \left( {\bf{A}} \right)\Delta \left( {\bf{B}} \right)\prod\nolimits_{i = 1}^{{N_r}} {{{\left( {s + i - 2} \right)}^{i - 1}}} }} \\
&\times\det \left( \begin{array}{l}
{\left\{ {\sum\limits_{k = 0}^\infty  {\frac{{\Gamma \left( { - s - {N_r} - k + 1} \right)}}{{\Gamma \left( { - s - {N_r} + 2} \right)}}{{\left( {-\frac{{{a_i}{b_j}}}{\rho }} \right)}^k}} } \right\}_{1 \le i \le {N_r},j}}\\
{\left\{ {{b_j}^{{N_t} - i}} \right\}_{{N_r} + 1 \le i \le {N_t},j}}
\end{array} \right)
\end{aligned}.$}
\end{equation}
By using eqs. (8.380.1), (8.384.1), and (8.384.3) in \cite{gradshteyn1965table}, the infinite series in \eqref{eqn:mel_fc_defa_sub} can be rewritten as
\begin{equation}\label{eqn:inte_tric}
\resizebox{0.9\hsize}{!}{$
\begin{aligned}
&\sum\limits_{k = 0}^\infty  {\frac{{\Gamma \left( { - s - {N_r} - k + 1} \right)}}{{\Gamma \left( { - s - {N_r} + 2} \right)}}{{\left( { - \frac{{{a_i}{b_j}}}{\rho }} \right)}^k}}\\ &=\int_0^\infty  {{{\left( {1 + y} \right)}^{s + {N_r} - 2}}{e^{ - \frac{{{a_i}{b_j}}}{\rho }y}}dy}=\Psi \left( {1,s + {N_r};\frac{{{a_i}{b_j}}}{\rho }} \right)
\end{aligned},$}
\end{equation}
where the last equality holds by using  (9.211.4) in \cite{gradshteyn1965table}, $\Psi \left( \cdot,\cdot;\cdot \right)$ denotes  confluent hypergeometric function.
By substituting \eqref{eqn:inte_tric} into (\ref{eqn:mel_fc_defa_sub}), we finally arrive at \eqref{eqn:mellin_fullc_tricomi}.
\vspace{-1em}
\section{}\label{app:full_cdf}
By applying the determinant expansion to \eqref{eqn:mellin_fullc_tricomi}, we get
\begin{equation}\label{eqn:cdf_G_fullc}
\resizebox{0.85\hsize}{!}{$\begin{aligned}
&{F_G}\left( x \right) = \frac{1}{{2\pi {\rm{i}}}}\frac{{\Lambda\prod\nolimits_{i = 1}^{N_r} {{a_i}^{N_r}} \prod\nolimits_{j = 1}^{N_t} {{b_j}^{N_r}} }}{{\rho^{\frac{1}{2}{{{N_r}\left( {{N_r} + 1} \right)}}}}}\sum\limits_{{\bs{\sigma }} \in {S_{N_t}}} {{\mathop{\rm sgn}} \left( {\bs{\sigma }} \right)}\\
&\times {\prod\limits_{i = {N_r} + 1}^{N_t} {{b_{{\sigma _i}}}^{{N_t} - i}} }\int\nolimits_{c - {\rm i}\infty }^{c + {\rm i}\infty } \frac{{\Gamma \left( { - s} \right)}}{{\Gamma \left( {1 - s} \right)}}\prod\limits_{j = 1}^{{N_r}-1} \frac{{\Gamma \left( { - s - {N_r} + 1} \right)}}{{\Gamma \left( { - s - j + 1} \right)}}\\
&\times\prod\limits_{i = 1}^{N_r}\Psi \left( {1,s + 1 + {N_r};\frac{{{a_{{i}}}{b_{\sigma _i}}}}{\rho }} \right){x^{ - s}} ds
\end{aligned}.$}
\end{equation}
By using the definition of $\Xi \left( {a,\alpha ,A,\varphi } \right) = {A^{\varphi  + a + \alpha s - 1}}\Psi \left( {\varphi ,\varphi  + a + \alpha s;A} \right)$ in \cite{shi2017asymptotic} and identifying the integration in \eqref{eqn:cdf_G_fullc} with the generalized Fox's H function \cite{yilmaz2010outage}, \eqref{eqn:cdf_G_fullc} can finally be represented as (\ref{eqn:F_G_CDF_fullcorr}).
\vspace{-0.5em}
\section{}\label{app:asy_y_3}
Applying property 2 in \cite{shi2017asymptotic} to (\ref{eqn:F_G_CDF_fullcorr}) gives rise to
\begin{equation}\label{eqn:F_G_CDF_fullcora_p2}
\resizebox{0.8\hsize}{!}{$\begin{aligned}
&{\mathcal Y_{ {{\bs\sigma}}}(x)}={{{{\rho }}}^{-{N_r}^2}}\prod\limits_{i = 1}^{N_r} {{a_i}^{N_r}}{{b_{\sigma_i}}^{N_r}}\frac{1}{{2\pi {\rm{i}}}}\int\nolimits_{-c - {\rm i}\infty }^{-c + {\rm i}\infty } \frac{{\Gamma \left( s \right)}}{{\Gamma \left( {1 + s} \right)}}\\
&\times\prod\limits_{j = 1}^{{N_r}-1}{\frac{{\Gamma \left( {s - {N_r} + 1} \right)}}{{\Gamma \left( {s - j + 1} \right)}}}\prod\limits_{i = 1}^{N_r} {\Psi \left( {1,1 + {N_r} - s;{\frac{{{a_i}{b_{{\sigma _i}}}}}{\rho }}} \right)}{x^s} ds
\end{aligned}.$}
\end{equation}
We set $c < -{N_t}{N_r}+\frac{1}{2}{{N_r}({N_r}-1)}$. By using \cite[eq.(9.210.2)]{gradshteyn1965table} and ignoring the higher order term $o\left( {{\rho ^{ - {N_t}{N_r}-\frac{1}{2}N_r(N_r+1)}}} \right)$ in \eqref{eqn:F_G_CDF_fullcora_p2},
together with the series expansion as $\Psi\left( {\alpha ,\beta ;x} \right) = \sum\nolimits_{n = 0}^\infty  {{{{\left( \alpha  \right)}_n}}}{{{x^n}}}/{{{{\left( \beta  \right)}_n}}}/{{n!}}$  \cite[eq.(9.210.1)]{gradshteyn1965table}, where $(\cdot)_n$ is Pochhammer symbol, we finally arrive at \eqref{eqn:mellin_cdf_g_fullcsers_exp}.
\vspace{-1em}
\section{}\label{app:asy_out_full}
Substituting \eqref{eqn:mellin_cdf_g_fullcsers_exp} into \eqref{eqn:F_G_CDF_fullcorr} and swapping the orders of integration and summation produces, after some basic algebraic manipulations, (\ref{eqn:F_G_CDF_fullcorr}) can be further derived as
\begin{equation}
\begin{aligned}\label{eqn:F_G_mellin_trans_meijerg}
p_{out}=& \frac{\Lambda\prod\nolimits_{i = 1}^{N_r} {{a_i}^{N_r}} \prod\nolimits_{j = 1}^{N_t} {{b_j}^{N_r}}}{{\rho }^{\frac{1}{2}{{N_r}\left( {{N_r} + 1} \right)}}}\zeta(N_r)\xi(N_t) + o\left( {\rho ^{ - {N_t}{N_r}}} \right)
\end{aligned}.
\end{equation}
where \resizebox{0.8\hsize}{!}{$\xi(N_t)=\sum\limits_{{\bs{\sigma }} \in {S_{N_t}}} {{\mathop{\rm sgn}} \left( {\bs{\sigma }} \right)}\prod\limits_{i = 1}^{N_r} {{{\left( { - \frac{{{a_{{ i}}}{b_{\sigma_i}}}}{\rho }} \right)}^{{n_{i}}}}} \prod\limits_{i = {N_r} + 1}^{N_t} {{b_{\sigma_i}}^{{N_t} - { i}}}$} and \resizebox{0.9\hsize}{!}{$\zeta(N_r)=\sum\limits_{{n_1}, \cdots ,{n_{N_r}} = 0}^\infty  {{\frac{1}{{2\pi {\rm{i}}}}\int\nolimits_{-c - {\rm i}\infty }^{-c + {\rm i}\infty } {\frac{{\Gamma \left( s \right)}}{{\Gamma \left( {1 + s} \right)}}\prod\limits_{i = 1}^{N_r} {\frac{{\Gamma \left( {s - {N_r} - {n_i}} \right)}}{{\Gamma \left( {s - i + 1} \right)}}}  {2^{Rs}}ds} }}$.} By using the determinant expansion in $\xi(N_t)$ and expressing $\zeta(N_r)$ in terms of Meijer G-function, leads to
\begin{equation}\label{eqn:F_G_mellin_trans_det_2}
\begin{aligned}\xi(N_t)
&={\prod\limits_{i = 1}^{N_r} {{{\left( { - \frac{{{a_i}}}{\rho }} \right)}^{{n_i}}}} }\det \left( {\begin{array}{*{20}{c}}
{{{\left\{ {{b_j}^{{n_i}}} \right\}}_{1 \le i \le {N_r},j}}}\\
{{{\left\{ {{b_j}^{{N_t} - i}} \right\}}_{{N_r} + 1 \le i \le {N_t},j}}}
\end{array}} \right)\\
&= \prod\limits_{i = 1}^{{N_r}} {{{\left( { - \frac{{{a_i}}}{\rho }} \right)}^{{n_i}}}} {\rm{sgn}}\left( {\bf{n}} \right){( - 1)^{{N_r}({N_t} - {N_r})}}\Delta \left( {\bf{B}} \right)
\end{aligned},
\end{equation}
\begin{equation}\label{eqn:F_G_mellin_trans_det_1}
\resizebox{1\hsize}{!}{$\zeta(N_r) =\sum\limits_{{\bf{n}} \in {\mathbb {N} ^{N_r}}} {G_{{N_r} + 1,{N_r} + 1}^{0,{N_r} + 1}\left( {\left. {\begin{array}{*{20}{c}}
{1,1 + {N_r} + {n_{\rm{1}}}, \cdots ,1 + {N_r} + {n_{N_r}}}\\
{0,1, \cdots ,{N_r}}
\end{array}} \right|2^R} \right)}$},
\end{equation}
where ${\bf n} = (n_1,\cdots,n_{N_r})$. Notice that any term with $n_i=0,\cdots,{N_t}-{N_r}-1$ is equal to zero thanks to the basic property of determinant, the dominant terms with $\bf n$ belonging to the set of the permutations of ${\Omega _{N_r}} = \left\{ {{N_t} - {N_r}, \cdots ,{N_t} - 1} \right\}$ can produce non-zero determinants. Substituting (\ref{eqn:F_G_mellin_trans_det_1}) and (\ref{eqn:F_G_mellin_trans_det_2}) into (\ref{eqn:F_G_mellin_trans_meijerg}), and ignoring the terms with both zero value of the determinant and the order of $\rho$ larger than $N_tN_r$, $p_{out}$ can be asymptotically expanded as
\begin{equation}\label{eqn:F_G_mellin_trans_detzero}
\resizebox{0.8\hsize}{!}{$\begin{aligned}
p_{out} =&
 \frac{{\left( { - 1} \right)}^{\frac{1}{2}{{{N_r}\left( {{N_r} - 1} \right)}}}\prod\nolimits_{i = 1}^{N_r} {{a_i}^{N_r}} \prod\nolimits_{j = 1}^{N_t} {{b_j}^{N_r}} }{{{{{\rho }} }^{{N_t}{N_r}}}{\Delta \left( {\bf{A}} \right)}}\\
&\times G_{{N_r} + 1,{N_r} + 1}^{0,{N_r} + 1}\left( {\left. {\begin{array}{*{20}{c}}
{1,{N_t} + 1, \cdots ,{N_t} + {N_r}}\\
{0,1, \cdots ,{N_r}}
\end{array}} \right|2^R} \right)\\
&\times \sum\nolimits_{{\bf{n}} \in {\Omega _{N_r}}} {{\mathop{\rm sgn}} \left( {{\bf n}} \right)\prod\limits_{i = 1}^{N_r} {{a_i}^{{n_i}}} }  + o\left( {{\rho ^{ - {N_t}{N_r}}}} \right),
\end{aligned}$}
\end{equation}
where the Meijer-G function can be extracted from the summation as a common factor due to the fact that its value is independent of the order of the elements of $\bf n$. Since the following equality holds
\begin{equation}\label{eqn:detea_identi}
\sum\nolimits_{{\bf{n}} \in {\Omega _{N_r}}} {{\mathop{\rm sgn}} \left( {{\bf n}} \right)\prod\limits_{i = 1}^{N_r} {{a_i}^{{n_i}}} }= {\left( { - 1} \right)}^{\frac{1}{2}{{{N_r}\left( {{N_r} - 1} \right)}}}\prod\limits_{i = 1}^{N_r} {{a_i}^{{N_t} - {N_r}}} \Delta \left( {\bf{A}} \right),
\end{equation}
the asymptotic $p_{out}$ can be finally derived as \eqref{eqn:F_G_mellin_trans_detzerofina}.
\vspace{-1em}
\bibliographystyle{ieeetran}
\bibliography{mimo}

\end{document}